\documentclass[runningheads,a4paper]{llncs}

\usepackage{amssymb}
\usepackage{amsmath}
\usepackage{graphicx}

\usepackage{epsfig}
\usepackage{graphicx}
\usepackage{dcolumn}
\usepackage{bm}
\usepackage{tikz}
\usepackage[all]{xy}
\usetikzlibrary{arrows,shapes}
\usepackage{url}
\usepackage{centernot}
\usepackage{stmaryrd}

\begin{document}

\mainmatter  

\title{Path homology and temporal networks
}

\titlerunning{Path homology and temporal networks}

%
%
\author{Samir Chowdhury\inst{1},
Steve Huntsman\inst{2}, 
Matvey Yutin\inst{3}
}
\authorrunning{S. Chowdhury, S. Huntsman and M. Yutin}

\institute{
Stanford University 
(\email{samirc@stanford.edu})
\and
BAE Systems FAST Labs 
\and
University of California, San Diego 
(\email{myutin@ucsd.edu})
}

%
%

\toctitle{Path homology and temporal networks}
\tocauthor{S. Chowdhury and S. Huntsman and M. Yutin}
\maketitle

\begin{abstract}
We present an algorithm to compute path homology for simple digraphs, and use it to topologically analyze various small digraphs \emph{en route} to an analysis of complex temporal networks which exhibit such digraphs as underlying motifs. 
The digraphs analyzed include all digraphs, directed acyclic graphs, and undirected graphs up to certain numbers of vertices, as well as some specially constructed cases. Using information from this analysis, we identify small digraphs contributing to path homology in dimension $2$ for three temporal networks, and relate these digraphs to network behavior. We conclude that path homology can provide insight into temporal network structure and \emph{vice versa}.
\end{abstract}

\section{\label{sec:Introduction}Introduction}

Path homology of digraphs (briefly recalled in \S \ref{sec:PathHomology}) was introduced in \cite{Grigoryan2012} and concurrent papers \cite{Grigoryan2014,Grigoryan2014b,Grigoryan2015,Grigoryan2017,Grigoryan2018,Grigoryan2018b}, and it is increasingly being studied by practitioners of applied topology \cite{Chowdhury2018,phmlp,dey2020efficient,huntsman-cyclo,lin2019weighted}. However, progress in empirical applications of path homology to digraph-structured data in general, especially large and/or complex digraphs, is obstructed by the difficulty of relating the algebraic information in homology to the combinatorial structure of the underlying digraph. 

To help remedy this, in \S \ref{sec:Algorithm} we present an algorithm and implementation for computing path homology in arbitrary dimension with
capabilities for symbolic output. We then use this algorithm in \S \ref{sec:Phenomenology} to compute the path homologies of digraphs from the following families: (1) all digraphs on $\leq 4$ vertices, (2) all directed acyclic graphs on $\leq 6$ vertices, (3) all undirected graphs on $\leq 6$ vertices, (4) Erd\H{o}s-R\'{e}nyi random graphs, and (5) a family of graphs conjectured to exhibit torsion. These examples build intuition and provide digraph constructions with surprising path homologies. 

Finally, we use this algorithm in \S \ref{sec:TemporalNetworks} to compute path homologies of digraphs obtained from three temporal networks, identifying salient motifs of the sort found in \S \ref{sec:Phenomenology} and relating them to the broader network behavior, before making concluding remarks in \S \ref{sec:Remarks}.

\section{\label{sec:PathHomology}Path homology}\label{Path}

We outline \emph{path homology}\index{homology!path} as treated in \cite{Grigoryan2012,Chowdhury2018}. 
For additional background on topology, including the theory of simplicial homology that shares many similarities with path homology, see \cite{GhristEAT,Hatcher}.

For a loopless digraph $D = (V,A)$, the set $\mathcal{A}_p(D)$ of \emph{allowed $p$-paths} is
\begin{equation}
\label{eq:allowed}
\{(v_0,\dots,v_p) \in V^{p+1} : (v_{j-1},v_j) \in A, 1 \le j \le p\}.
\end{equation}
By convention, we take $\mathcal{A}_{0} := V$, $V^0 \equiv \mathcal{A}_{-1} := \{0\}$ and $V^{-1} \equiv \mathcal{A}_{-2} := \varnothing$. For a field $\mathbb{F}$
\footnote{
Path homology can also be defined over rings such as $\mathbb{Z}$ without any modifications besides a change of notation. This definition gives additional power: M. Yutin has exhibited digraphs on as few as six vertices that have torsion: see Fig. \ref{fig:torsion}.
}
and a finite set $X$, let $\mathbb{F}^X \cong \mathbb{F}^{|X|}$ be the free $\mathbb{F}$-vector space on $X$, taking $\mathbb{F}^\varnothing := \{0\}$. The \emph{non-regular boundary operator} $\partial_{[p]} : \mathbb{F}^{V^{p+1}} \rightarrow \mathbb{F}^{V^p}$ is the linear map acting on the standard basis as
\begin{equation}
\label{eq:preboundary}
\partial_{[p]} e_{(v_0,\dots,v_p)} = \sum_{j=0}^p (-1)^j e_{\nabla_j (v_0,\dots,v_p )},
\end{equation}
where $\nabla_j (v_0,\dots,v_p )$ notation means that $v_j$ is dropped from the tuple.
A few lines of algebra suffice to confirm that $\partial_{[p-1]} \circ \partial_{[p]} \equiv 0$, so $(\mathbb{F}^{V^{p+1}},\partial_{[p]})$ is a \emph{chain complex}, i.e., $\partial_{p-1} \circ \partial_p \equiv 0$, as indicated in Fig. \ref{fig:ChainComplex}. 

This is important because any chain complex $(C_p,\partial_p)$ gives rise to an invariant called \emph{homology} which behaves nicely with respect to maps on the chain complex induced by an underlying transformation of a common structure (here, a digraph). Writing $Z_p := \text{ker } \partial_p$ and $B_p := \text{im } \partial_{p+1}$, the dimension $p$ homology of the chain complex $(C_p,\partial_p)$ is the quotient
\begin{equation}
H_p := Z_p/B_p.
\end{equation}
This is a finitely generated abelian group, thus of the form $\mathbb{Z}^{\beta_p} \oplus T_p$, where the \emph{torsion} $T_p$ is a finite abelian group. Over a field, the $C_p$ are vector spaces and the torsion vanishes, in which case the \emph{Betti numbers} $\beta_p := \text{dim } H_p = \text{dim } Z_p - \text{dim } B_p$ fully capture the homology up to isomorphism. 

\begin{figure}[htb]
\begin{center}
\includegraphics[trim = 60mm 114mm 60mm 116mm, clip, keepaspectratio, width=.67\textwidth]{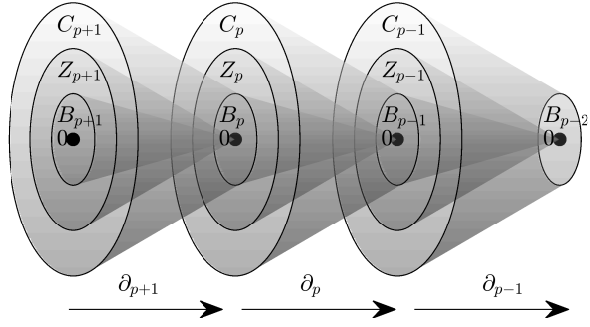}
\caption[Schematic of a chain complex]{Schematic of a chain complex. Here $C_p$, $B_{p-1}$, and $Z_p$ are respectively the domain, codomain, and kernel of $\partial_p$, so that the homology $H_p := Z_p/B_p$ is well defined.}
\label{fig:ChainComplex}
\end{center}
\end{figure}

Path homology comes from a chain complex derived from $(\mathbb{F}^{V^{p+1}},\partial_{[p]})$. Set
\begin{equation}
\label{eq:invariant}
\Omega_p := \left \{ \omega \in \mathbb{F}^{\mathcal{A}_{p}} : \partial_{[p]} \omega \in \mathbb{F}^{\mathcal{A}_{p-1}} \right \},
\end{equation}
$\Omega_{-1} := \mathbb{F}^{\{0\}} \cong \mathbb{F}$, and $\Omega_{-2} := \mathbb{F}^{\varnothing} = \{0\}$. We have that $\partial_{[p]} \Omega_p \subseteq \mathbb{F}^{\mathcal{A}_{p-1}}$, so $\partial_{[p-1]} \partial_{[p]} \Omega_p = 0 \in \mathbb{F}^{\mathcal{A}_{p-2}}$ and $\partial_{[p]} \Omega_p \subseteq \Omega_{p-1}$. The \emph{(non-regular) path complex} of $D$ is accordingly defined as the chain complex $(\Omega_p,\partial_p)$, where $\partial_p := \partial_{[p]}|_{\Omega_p}$.
\footnote{
The implied \emph{regular path complex} prevents a directed 2-cycle from having nontrivial 1-homology. While \cite{Grigoryan2012} advocates regular path homology, in our view non-regular path homology is simpler, richer, and more likely useful in applications. Our rationale amounts to the convention that a directed 2-cycle should count as a ``hole.''
}
The homology of this path complex is the \emph{(non-regular) path homology} of $D$. 

As a basic example illustrating the mechanics of path homology, consider the digraphs $D_1$ and $D_2$ in Figure \ref{fig:phDigraphs}.  $\mathcal{A}_1(D_1)$ and $\mathcal{A}_1(D_2)$ are given by the digraph arcs, $\mathcal{A}_2(D_2) = \varnothing$, and $\mathcal{A}_2(D_2) = \{(w,x,z), (w,y,z)\}$. Now $\partial_{[2]}e_{(w,x,z)} = e_{(x,z)} - e_{(w,z)} + e_{(w,x)} \not\in \mathbb{F}^{\mathcal{A}_1(D_2)}$ and $\partial_{[2]}e_{(w,y,z)} = e_{(y,z)} - e_{(w,z)} + e_{(w,z)} \not\in \mathbb{F}^{\mathcal{A}_1(D_2)}$ (because the arc $w \rightarrow z$ is missing), so
\begin{align}
\partial_{[2]}(e_{(w,x,z)} - e_{(w,y,z)}) & = e_{(x,z)} - e_{(w,z)} + e_{(w,x)} - e_{(y,z)} + e_{(w,z)} - e_{(w,y)} \nonumber \\ 
& = e_{(x,z)} + e_{(w,x)} - e_{(y,z)} - e_{(w,y)} \in \mathbb{F}^{\mathcal{A}_1(D_2)}. \nonumber
\end{align}
Therefore the dimensions of the path homology vector spaces--i.e., the Betti numbers--are different: $\beta_1(D_1) = 1$ and $\beta_1(D_2) = 0$. 

\begin{figure}[htb]
\centering
\begin{tikzpicture}[->,>=stealth',shorten >=1pt,scale = 1,every node/.style={transform shape}]]
	\node (D1) at (.5,1.5) {$D_1$};
	\node [circle, draw, align=center] (w) at (0,1) {$1$};
	\node [circle, draw, align=center] (x) at (1,1) {$2$};
	\node [circle, draw, align=center] (y) at (0,0) {$3$};
	\node [circle, draw, align=center] (z) at (1,0) {$4$};
	\draw (w) to (x);
	\draw (w) to (y);
	\draw (z) to (x);
	\draw (z) to (y);
\end{tikzpicture}
\quad \quad \quad
\begin{tikzpicture}[->,>=stealth',shorten >=1pt,scale = 1,every node/.style={transform shape}]]
	\node (D2) at (.5,1.5) {$D_2$};
	\node [circle, draw, align=center] (w) at (0,1) {$1$};
	\node [circle, draw, align=center] (x) at (1,1) {$2$};
	\node [circle, draw, align=center] (y) at (0,0) {$3$};
	\node [circle, draw, align=center] (z) at (1,0) {$4$};
	\draw (w) to (x);
	\draw (w) to (y);
	\draw (x) to (z);
	\draw (y) to (z);
\end{tikzpicture}
\caption[Similar digraphs with different path homologies]{The digraph $D_1$ has trivial path homology but the digraph $D_2$ does not.}
\label{fig:phDigraphs}
\end{figure}
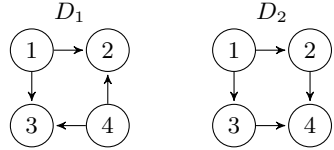

For convenience, we replace the path complex $(\Omega_p,\partial_p)$ with its \emph{reduction}
\begin{equation}
\label{eq:reduction}
\dots \Omega_{p+1} \overset{\partial_{p+1}}{\longrightarrow} \Omega_p \overset{\partial_p}{\longrightarrow} \Omega_{p-1} \overset{\partial_{p-1}}{\longrightarrow} \dots \overset{\partial_1}{\longrightarrow} \Omega_0 \overset{\tilde \partial_0}{\longrightarrow} \mathbb{F} \longrightarrow 0
\end{equation}
which (using an obvious notational device and assuming the original complex is nondegenerate) has the minor effect $\tilde H_0 \oplus \mathbb{F} \cong H_0$, while $\tilde H_p \cong H_p$ for $p > 0$. Similarly, $\tilde \beta_p = \beta_p - \delta_{p0}$, where $\delta_{jk} := 1$ iff $j = k$ and $\delta_{jk} := 0$ otherwise.

\subsection{\label{sec:Dyad}Path homologies of the mutual dyad motifs}

As a further example of practical relevance, we characterize the path homologies of a family of network motifs that we call the \emph{$n$-uplinked mutual dyads}---or dually, the $n$-downlinked mutual dyads---in reference to the original terminology going back to \cite{Milo}. Given an integer $n \geq 1$, an $n$-uplinked mutual dyad is a digraph $W_n$ with vertex set $\{a,b,1,2,\ldots,n\}$ and edge set $\{(a,b),(b,a)\}\cup \{(a,i):1\leq i \leq n\} \cup  \{(b,i):1\leq i \leq n\}$. This is illustrated in Fig. \ref{fig:EmailMotif}. The $n$-downlinked mutual dyad is defined by reversing all the arrows (cf. Fig. \ref{fig:MathOverflow1}). We have:
\begin{proposition} \label{prop:dyad}
Let $n \in \mathbb{Z}_{>0}$ and let $W_n$ denote the $n$-uplinked or downlinked mutual dyad. Then $\tilde \beta_2(W_n) = n-1$, and $\tilde \beta_p(W_n) = 0$ for all $p \geq 0, p\neq 2$.
\end{proposition}
\begin{proof} Suppose first that $W_n$ is the uplinked mutual dyad motif. From \cite{Grigoryan2012} we know that $\beta_0$ counts the number of connected components of the underlying undirected graph, so $\beta_0=1$ and $\tilde{\beta}_0=0$. Next consider the case $p=1$. We have $\partial_{[1]}(e_{(a,b)} + e_{(b,a)}) = 0$, but we also have $\partial_{[2]}(e_{(a,b,1)} + e_{(b,a,1)}) = e_{(a,b)} + e_{(b,a)}$, and so $e_{(a,b)} + e_{(b,a)}$ cannot contribute to $\tilde{\beta}_1$. Similarly terms of the form $e_{(a,b)} + e_{(b,i)} - e_{(a,i)} \in Z_1, \, 1\leq i \leq n,$ cannot contribute to $\tilde{\beta}_1$ as they belong to $B_1$, being the images of $e_{(a,b,i)}$ for $1\leq i \leq n$. Next consider $p\geq 3$. In these cases we find $\Omega_p = \{0\}$, as all the 3-paths have boundaries with non-allowed paths, and taking linear combinations does not cancel out these non-allowed paths. 

Finally we deal with the case $p=2$. First let $1\leq i\neq j  \leq n$. Then $\partial_{[2]}(e_{(a,b,i)} + e_{(b,a,i)} - e_{(a,b,j)} - e_{(b,a,j)}) = e_{(a,b)} + e_{(b,a)} - e_{(a,b)} - e_{(b,a)} = 0$. Thus all 2-paths of the form $e_{(a,b,i)} + e_{(b,a,i)} - e_{(a,b,j)} - e_{(b,a,j)}$ belong to $Z_2$, but not to $B_2$ as $\Omega_3$ is trivial. Some linear algebra shows that a basis for $Z_2$ is given by the collection $\{ e_{(a,b,1)} + e_{(b,a,1)} - e_{(a,b,j)} - e_{(b,a,j)} : 2\leq j \leq n\}$. It follows that $\tilde{\beta}_2 = n-1$. To conclude the proof, observe that all of these arguments hold for the downlinked mutual dyad by replacing terms of the form $e_{(a,b,i)}$ with $e_{(i,a,b)}$.  \end{proof}

\section{\label{sec:Algorithm}Algorithm}

Although an algorithm to compute non-regular path homology is fairly obvious, our implementation \cite{ph-code} is apparently among the first for dimension $> 1$. Though \cite{Shajii2013,Slawinski2019} are considerably older, we were unaware of these for some time and could not find any published work drawing on them, so we outline our approach here.

First, to reduce computation, we remove all nonbranching limbs (i.e. chains of vertices of total degree $2$ terminating in leaves of degree $1$), since these do not affect homology (Theorem 5.1 of \cite{Grigoryan2012}). Similarly (see Prop. 3.25 of \cite{Grigoryan2012}), we break the graph into weakly connected components and compute homology componentwise. For each component $D$, we extend an order on vertices $V(D) \equiv [n]$ to the lexicographical ordering on paths. We inductively construct $\mathcal{A} _p(D)$ for \(0 \le p \le p _{\max}\) as follows: \(\mathcal{A} _0(D) = V(D)\), and $\mathcal{A} _p(D)$ is constructed by appending to every path in $\mathcal{A} _{p - 1}(D)$ every vertex that has an arc from the path's terminal vertex. The paths are constructed in lexicographical order for each $p$.

From here, we compute the indices (using a radix-$n$ expansion) that specify the inclusion $\mathcal{A}_p \hookrightarrow V^{p+1}$ under lexicographical ordering. 
We then construct the matrix representation $\partial_{[p,\mathcal{A}]}$ of the restriction of $\partial_{[p]}$ to $\mathbb{F}^{\mathcal{A}_p}$ using the standard basis. Let $\nabla_{[p,\mathcal{A}]}$ be the projection of $\partial _{[p, \mathcal{A}]}$ onto $\mathbb{F} ^{V ^p \backslash \mathcal{A} _{p - 1}}$ (i.e., the matrix obtained by removing rows of $\partial_{[p,\mathcal{A}]}$ that correspond to elements of $\mathcal{A}_{p-1}$), and $\Delta _{[p, \mathcal{A}]}$ be the projection onto $\mathbb{F} ^{A _{p-1}}$. The kernel of $\nabla_{[p,\mathcal{A}]}$ is $\Omega_p$. In practice, we remove rows of $\nabla_{[p,\mathcal{A}]}$ that are identically zero before computing this kernel, which yields $\Omega_p$ much more efficiently.

With a matrix representation $\Omega_{[p,\mathcal{A}]}$ for the kernel above in hand, we compute the chain boundary operator \(\partial _p = \Omega _{[p - 1, \mathcal{A}]} ^{-1} \Delta _{[p, \mathcal{A}]} \Omega _{[p, \mathcal{A}]}\) (i.e., the projection of $\partial _{[p, \mathcal{A}]}$ onto $\mathbb{F} ^{\mathcal{A} _{p - 1}}$, projected onto and restricted to the invariant space). We then compute the homology of this chain complex, e.g. we use the rank-nullity theorem and compute the matrix ranks to obtain Betti numbers, or we compute representatives for the homology groups (albeit somewhat lacking in geometric meaning) as the cokernels of $[\ker \partial _{p}] ^{T} \partial _{p + 1}$, or take the Smith normal form of the boundary matrices to find any torsion over $\mathbb{Z}$ (cf. Fig. \ref{fig:torsion}).
\footnote{
Further optimizations are imaginable. We can pre-process the graph more, in accord with Theorem 5.7 of \cite{Grigoryan2012}, and deal with the fallout in low dimensions.
Instead of performing a singular value decomposition on rather large boundary matrices (\emph{en route} to computing the rank), we can recursively build up the invariant spaces from lower dimensions --- each sub-path of an invariant path is itself an invariant path (since we assume no loops). A simple algorithm for this latter approach might be to check every pair of paths in dimension $p$ against every vertex to see where we can append `triangles' and `squares' (terminology borrowed from \cite{Grigoryan2012}); while promising, this approach generates too many paths, and reducing it to a basis is computationally nontrivial. For low dimensions, we could also directly compute Betti numbers from the digraph itself (e.g. Proposition 3.24 of \cite{Grigoryan2012}).
}

\section{\label{sec:Phenomenology}Phenomenology of path homology for small digraphs}

\paragraph{Small digraphs.}

In Fig. \ref{fig:DGsOn4Vertices}, we illustrate digraphs on 4 vertices with nontrivial homology in dimensions greater than 1. Surprisingly, nontrivial homology arises even in dimension 3 with just 4 vertices. Additionally, the left panel of Fig. \ref{fig:DAGsAndUGsOn6Vertices} shows directed acyclic graphs (DAGs) on 6 vertices with nontrivial homology in dimension 2. Observations from these DAGs led us to formulate and prove a conjecture about the path homology of deep feedforward neural networks \cite{phmlp}. Finally, the right panel of Fig. \ref{fig:DAGsAndUGsOn6Vertices} shows undirected graphs (considered as digraphs) on 6 vertices with nontrivial homology in dimension 2, highlighting that path homology is relevant to the undirected case as well.

\begin{figure}[htbp]
\centering
\includegraphics[trim = 75mm 120mm 75mm 120mm, clip, width=.45\textwidth,keepaspectratio]{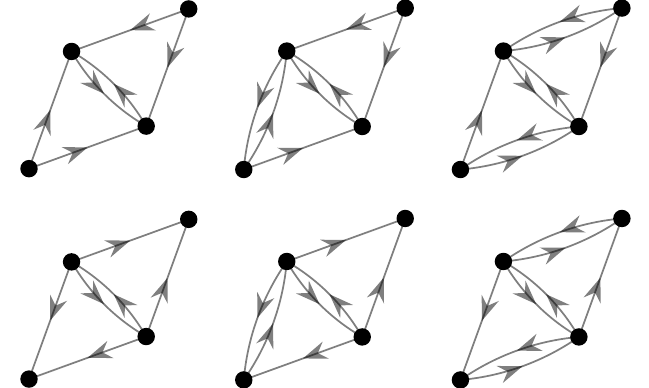}
\quad \quad
\includegraphics[trim = 75mm 120mm 75mm 120mm, clip, width=.45\textwidth,keepaspectratio]{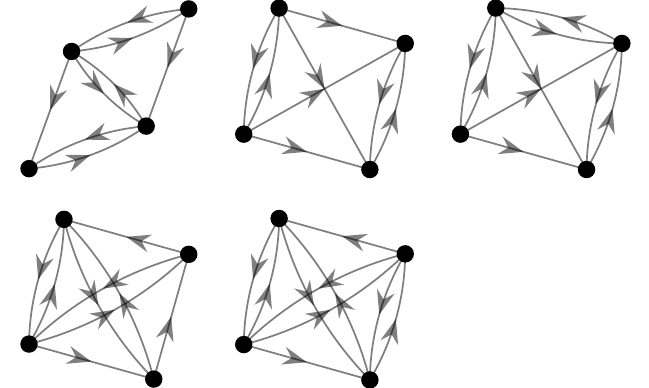}
\caption{
(L) $\tilde \beta_2 > 0$ for these 6 (of 218 total) digraphs on 4 vertices. In each case $\tilde \beta_p = \delta_{p,2}$. 
(R) $\tilde \beta_3 > 0$ for these 5 digraphs on 4 vertices. In each case $\tilde \beta_p = \delta_{p,3}$.
}
\label{fig:DGsOn4Vertices}
\end{figure} %

\begin{figure}[htbp]
\centering
\includegraphics[trim = 50mm 95mm 50mm 95mm, clip, width=.45\textwidth,keepaspectratio]{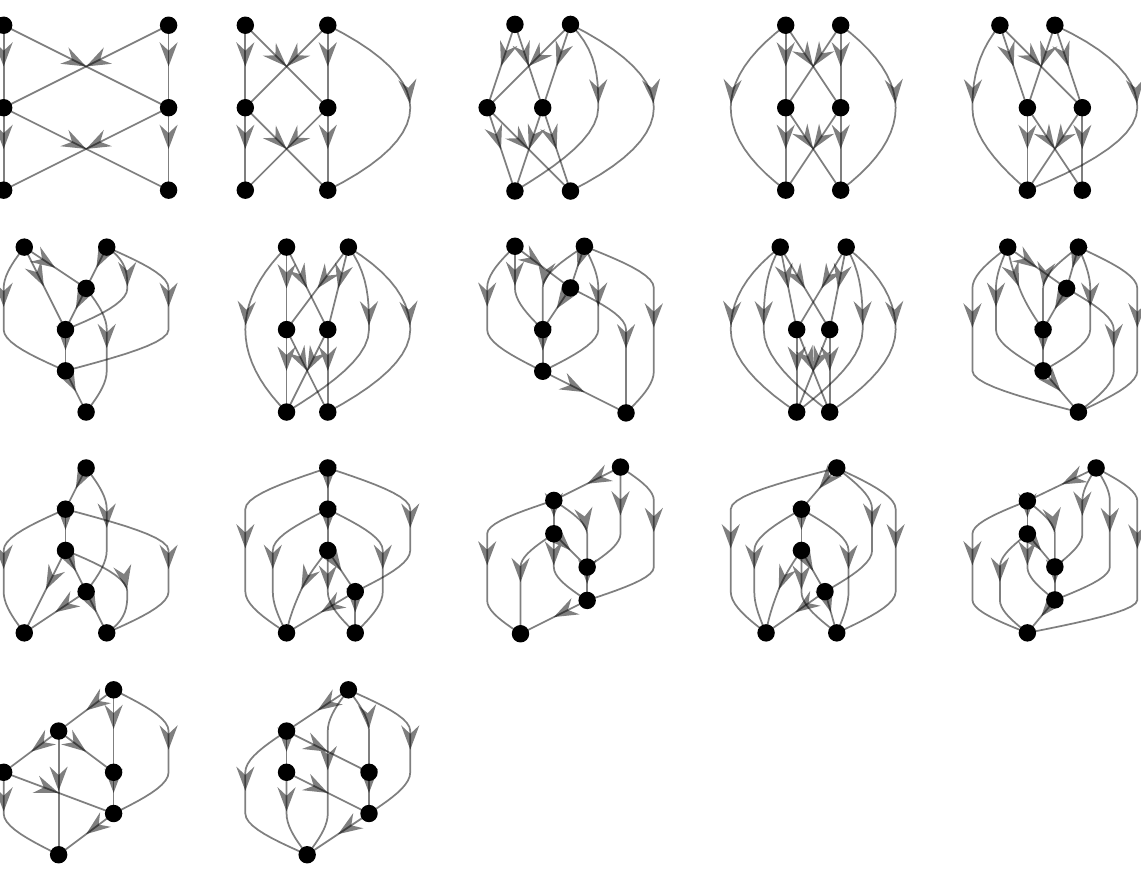}
\quad \quad
\includegraphics[trim = 50mm 95mm 50mm 95mm, clip, width=.45\textwidth,keepaspectratio]{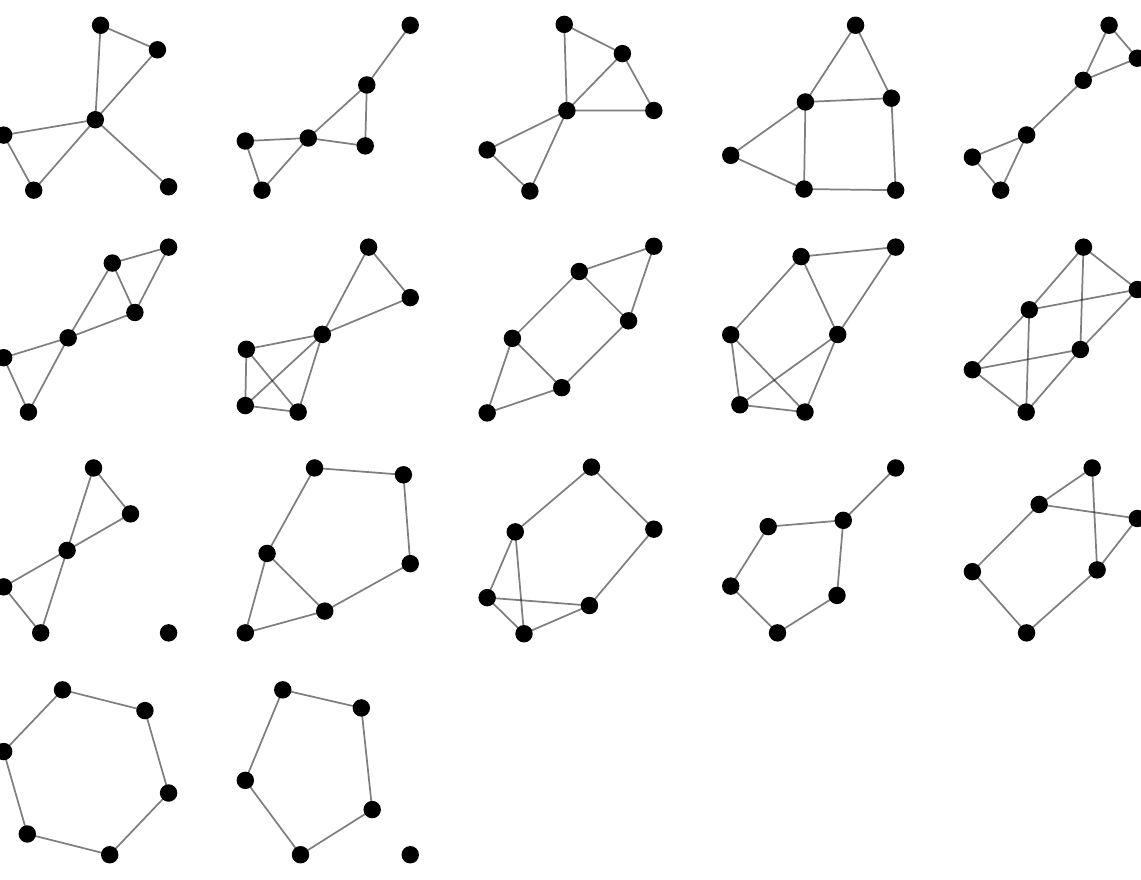}
\caption{
(L) $\tilde \beta_2 > 0$ for these 17 (of 5984 total) DAGs on 6 vertices. 
(R) $\tilde \beta_2 > 0$ for these 17 (of 156 total) undirected graphs on 6 vertices. 
}
\label{fig:DAGsAndUGsOn6Vertices}
\end{figure} %

\begin{figure}[htbp]
\centering
\includegraphics[trim = 95mm 130mm 90mm 130mm, clip, width=.25\textwidth,keepaspectratio]{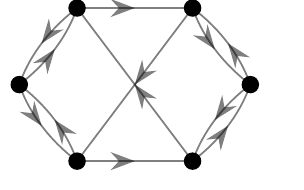}
\quad \quad
\includegraphics[trim = 95mm 130mm 90mm 130mm, clip, width=.25\textwidth,keepaspectratio]{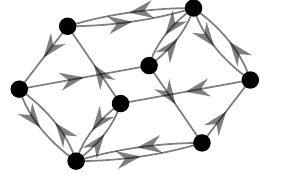}%
\caption{
These two digraphs exhibit torsion. Larger digraphs with torsion can be formed by adding more vertices to the central cycles of these and linking them each to one of the two vertices external to the cycle (in an alternating fashion). We conjecture that the digraph in this family with a central cycle of length $2 n$ has a torsion subgroup of $\mathbb{Z} / n \mathbb{Z}$ in $\tilde H_1$. This conjecture has been computationally verified up to $n = 8$. These digraphs may be analogues of so-called \emph{lens spaces} that can be formed by gluing two tori together with a twist, and that themselves exhibit similar torsion.
}
\label{fig:torsion}
\end{figure} %

\paragraph{Torsion.}

Though heretofore defined over fields, path homology still makes sense over rings, e.g. $\mathbb{Z}$. By sampling Erd\H{o}s-R\'{e}nyi digraphs, M. Yutin \cite{ph-torsion} was able to find a family of digraphs with torsion. In Fig. \ref{fig:torsion}, we show the smallest two digraphs in this family; larger digraphs are formed by adding more vertices to the central cycle and linking to one of the two vertices external to the cycle (in an alternating fashion). We conjecture that the digraph in this family with a central cycle of length $2n$ has a torsion subgroup of $\mathbb{Z} / n \mathbb{Z}$ in one-dimensional path homology. This conjecture has been computationally verified up to $n = 8$.


\paragraph{Erd\H{o}s-R\'{e}nyi random graphs.}

In Fig. \ref{fig:ER4} we plot empirical distributions for the first few Betti numbers of Erd\H{o}s-R\'{e}nyi random graphs \cite{FriezeKaronski} on 4 nodes. 
A standard application of these distributions could be to test if a stochastic digraph generating process could be modeled via a random graph generating model.


\begin{figure}[htbp]
\centering
\includegraphics[trim = 50mm 127mm 55mm 120mm, clip, width=0.8\textwidth
,keepaspectratio]{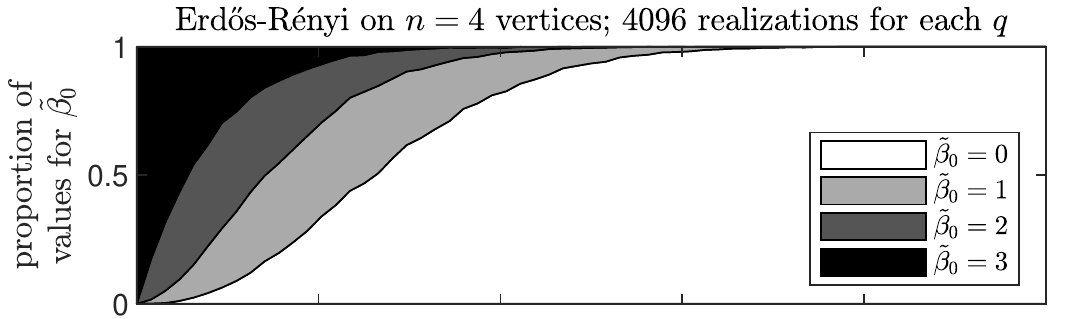}
\\
\includegraphics[trim = 50mm 127mm 55mm 123.9mm, clip, width=0.8\textwidth
,keepaspectratio]{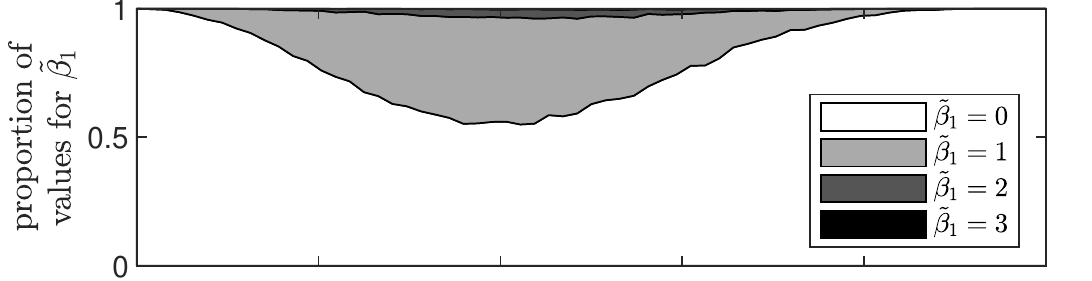}
\\
\includegraphics[trim = 50mm 127mm 55mm 123.9mm, clip, width=0.8\textwidth
,keepaspectratio]{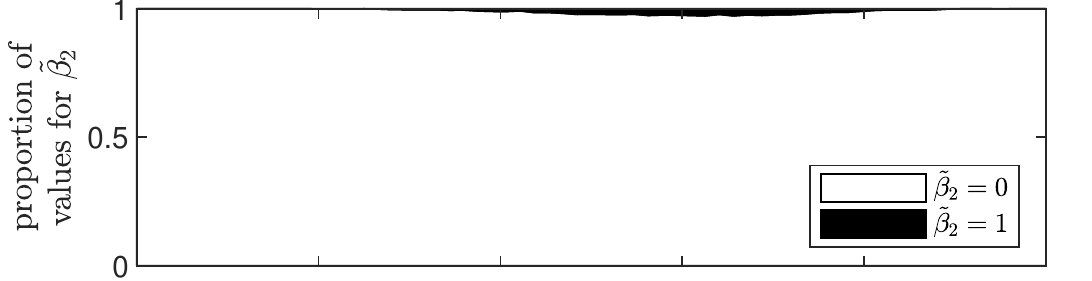}
\\
\includegraphics[trim = 50mm 120mm 55mm 123.9mm, clip, width=0.8\textwidth
,keepaspectratio]{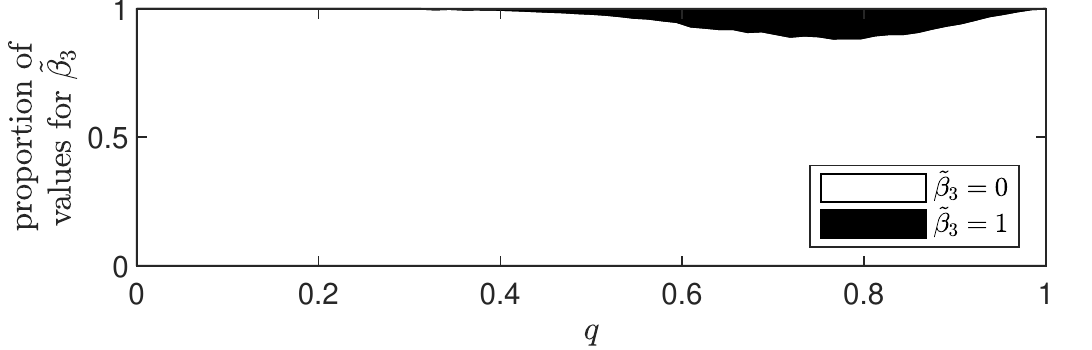}
\caption{ \label{fig:ER4} Empirical distributions of $\tilde \beta_p(D_{4,q})$.
} 
\end{figure}


\section{\label{sec:TemporalNetworks}Examples of applications to temporal networks}

We analyze three \emph{temporal networks} \cite{Holme,MasudaLambiotte}, more specifically \emph{directed contact networks} as described in \cite{CybenkoHuntsman}: MathOverflow, an email network, and activity on a Facebook group: these illustrate how path homology can find high-order interactions that are respectively indicators of dilution, recurring motifs, and concentration within network behavior.

\subsection{\label{sec:MathOverflow}MathOverflow}

To illustrate the ability of path homology to identify structurally relevant motifs, we analyzed the answer-to-question portion of the \texttt{sx-mathoverflow} temporal network available at \cite{Lescovec2014}. This network has 21688 vertices and 107581 directed temporal contacts, spanning 2350 days. It has previously been analyzed in \cite{MontoyaMaMondragon2013}; for a discussion of question/answer phenomenology on MathOverflow, see \cite{TausczikKitturKraut2014}.

We filtered the contacts through a sliding time window of 24 hours, moving every eight hours, and aggregated each window into a static digraph. We then computed the first three Betti numbers.
\footnote{
NB. Our path homology code removes any loops from digraphs.
}
Only two windows, immediately adjacent and overlapping, had $\beta_2 > 0$, corresponding to a period over 13-14 Oct 2009. Subsequent inspection of homology representatives revealed that this phenomenon originated in the presence of the 2-downlinked mutual dyad (cf. Sec. \ref{sec:Dyad}) motif presented in Fig. \ref{fig:MathOverflow1}, and additional inspection of MathOverflow itself revealed the particular questions and answers involved.


\begin{figure}[htbp]
\centering
\includegraphics[trim = 50mm 105mm 50mm 100mm, clip, width=0.75\textwidth,keepaspectratio]{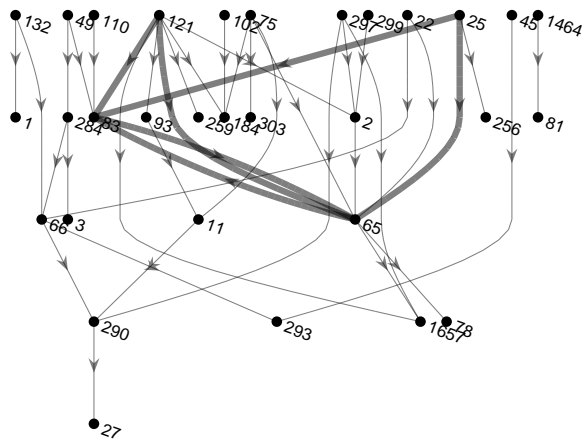}
\caption{ \label{fig:MathOverflow1} Digraph of activity on MathOverflow over a 24-hour period during 13-14 Oct 2009. Vertices are labeled by user ID; arcs are directed from answerer to questioner (any parallel arcs are merged). Arcs participating in a 2-homology representative are highlighted, with pairs of arcs and associated questions $((25,65),437)$, $((25,83),451)$, $((65,83),446)$, $((83,65),437)$, $((121,65),433)$, and $((121,83),446)$ as well as $((121,83),451)$. Three of the four users $\{25,65,83,121\}$ share the same first subject tag at time of writing; all four share the same second subject tag. 
} 
\end{figure}

The rarity of 2-homology is related to the fact that it happened very early in the history of MathOverflow--in fact, just two weeks after its beginning. As MathOverflow changed over time, opportunities for such tightly coupled patterns of questions and answers diminished. For example, most of the first 200 users asked and answered many fewer questions over time, while the overall size of and activity on MathOverflow grew much larger.

\subsection{\label{sec:Email}An email network} 

The phenomenon of the previous example is actually ubiquitous and generalized in email networks, for reasons attributable to well-known behaviors unique to the medium. We analyzed the \texttt{email-Eu-core-temporal} network available at \cite{Lescovec2014}. This network has 986 vertices and 332334 directed temporal contacts, spanning 804 days of activity. 
We filtered the contacts through a sliding window of the most recent 100 contacts, moving every 50 contacts, and again aggregated each window into a static digraph. We then computed the first three Betti numbers. Many windows exhibited very high values of $\beta_2$ due to instances of the $n$-uplinked mutual dyad (cf. Sec. \ref{sec:Dyad}) motif shown in Fig. \ref{fig:EmailMotif}. 

\begin{figure}[htbp]
\centering
	\includegraphics[trim = 45mm 110mm 45mm 110mm, clip, width=0.6\textwidth,keepaspectratio]{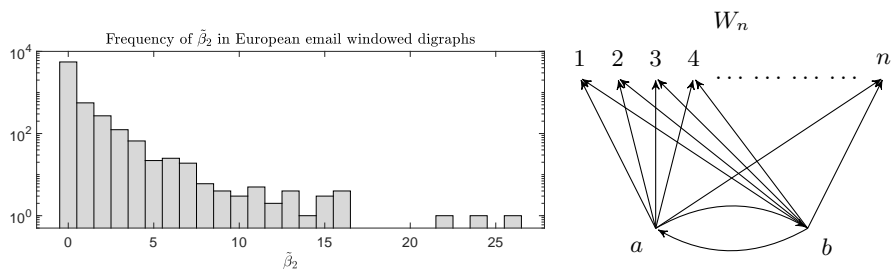}
	\ \
	\begin{tikzpicture}[every node/.style={inner sep=0,outer sep=0,scale=1},->,>=stealth',shorten >=1pt]
		\coordinate (vA) at (-1,0);
		\coordinate (vB) at (1,0);
		\coordinate (v1) at (-2,2);
		\coordinate (v2) at (-1.5,2);
		\coordinate (v3) at (-1,2);
		\coordinate (v4) at (-.5,2);
		\coordinate (v9) at (2,2);
		\draw (vA) to [out=30,in=150,looseness=1] (vB);
		\draw (vB) to [out=-150,in=-30,looseness=1] (vA);
		\foreach \from/\to in {
			vA/v1, vA/v2, vA/v3, vA/v4, vA/v9, vB/v1, vB/v2, vB/v3, vB/v4, vB/v9}
			\draw (\from) to (\to);
		\node (blank) at (0,-.67) {$$};
		\node (labelA) at (-1.25,-.25) {$a$};
		\node (labelB) at (1.25,-.25) {$b$};
		\node (label1) at (-2,2.25) {$1$};
		\node (label2) at (-1.5,2.25) {$2$};
		\node (label3) at (-1,2.25) {$3$};
		\node (label4) at (-.5,2.25) {$4$};
		\node (dots5) at (0,2) {$\dots$};
		\node (dots6) at (.5,2) {$\dots$};
		\node (dots7) at (1,2) {$\dots$};
		\node (dots8) at (1.5,2) {$\dots$};
		\node (label9) at (2,2.25) {$n$};
		\node (Wn) at (0,2.75) {$W_n$};
	\end{tikzpicture}
\caption{ \label{fig:EmailMotif} (L) The distribution of $\tilde \beta_2$ for windowed digraphs obtained from the email network. (R) The digraph $W_n$ depicted here has $\tilde \beta_2(W_n) = n-1$, and it is the cause of high values of $\tilde \beta_2$ in windowed digraphs obtained from the email network. The underlying dynamics is common in large organizations: two people (``Alice'' and ``Bob'') both send email to the same wide distribution and to each other.
} 
\end{figure}

\subsection{\label{sec:Facebook}A Facebook group}

As a final example, we consider the first 1000 days of activity on a Facebook group \cite{ViswanathEtAl,Kunegis}. Because the associated temporal network (13295 vertices; 187750 contacts) has a daily lull with virtually no activity, we aggregated the temporal network into daily digraphs. Fig. \ref{fig:FacebookStatistics} shows the number of posts per day and the first three Betti numbers. Besides the obvious correlation between activity and $\tilde \beta_0$, the appearance of progressively more- and higher-dimensional homology classes over time is also evident, indicating the emergence of higher-order network structure. Fig. \ref{fig:FacebookDigraph} shows the first daily digraph with $\tilde \beta_2 > 0$.

\section{\label{sec:Remarks}Remarks}

Although the computational requirements for path homology scale exponentially with dimension $p$, even the case $p = 2$ can highlight salient network structure and behavior. By decomposing temporal networks into time windows, path homology can be successfully brought to bear in this regard, illuminating both motifs with nontrivial path homology as well as the temporal networks themselves.

\begin{figure}[htbp]
	\includegraphics[trim = 70mm 110mm 70mm 110mm, clip, width=0.5\textwidth,keepaspectratio]{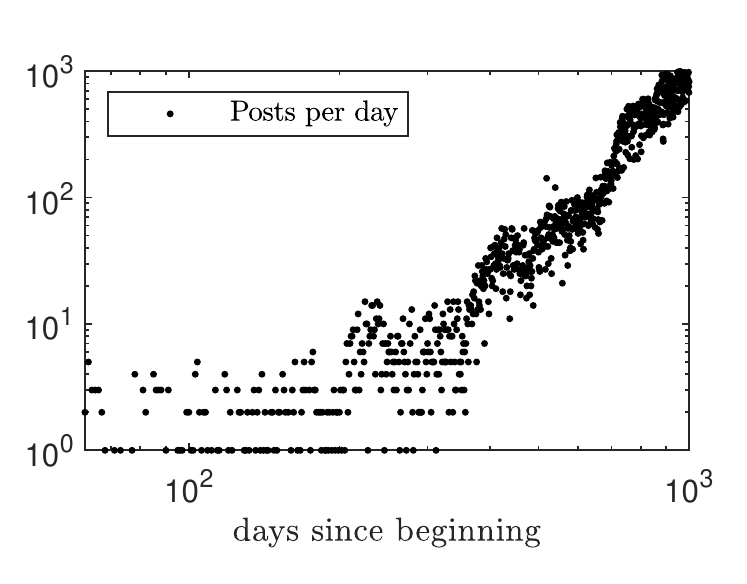}
	\includegraphics[trim = 70mm 110mm 70mm 110mm, clip, width=0.5\textwidth,keepaspectratio]{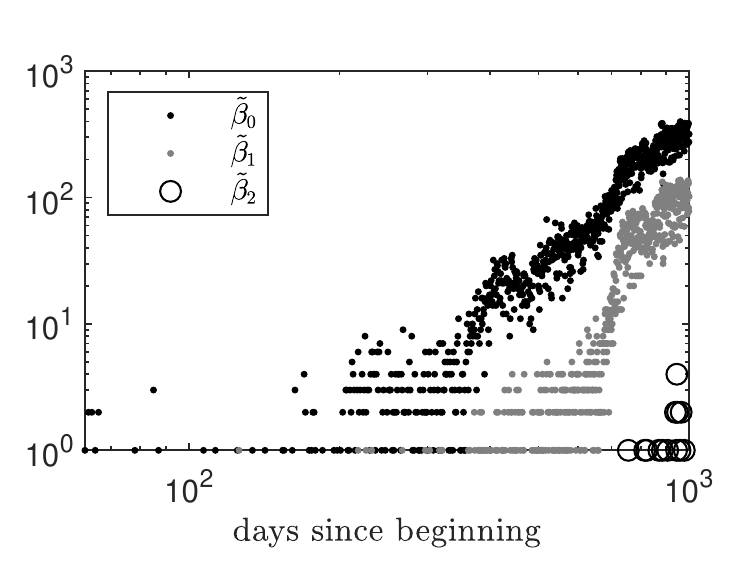}
\caption{ \label{fig:FacebookStatistics} (L) Daily Facebook group posts. (R) Betti numbers of daily digraphs. As activity increases, so do topological features in dimensions 0 through 2.
}
\end{figure}

\begin{figure}[htbp]
	\includegraphics[trim = 65mm 105mm 60mm 100mm, clip, width=0.5\textwidth,keepaspectratio]{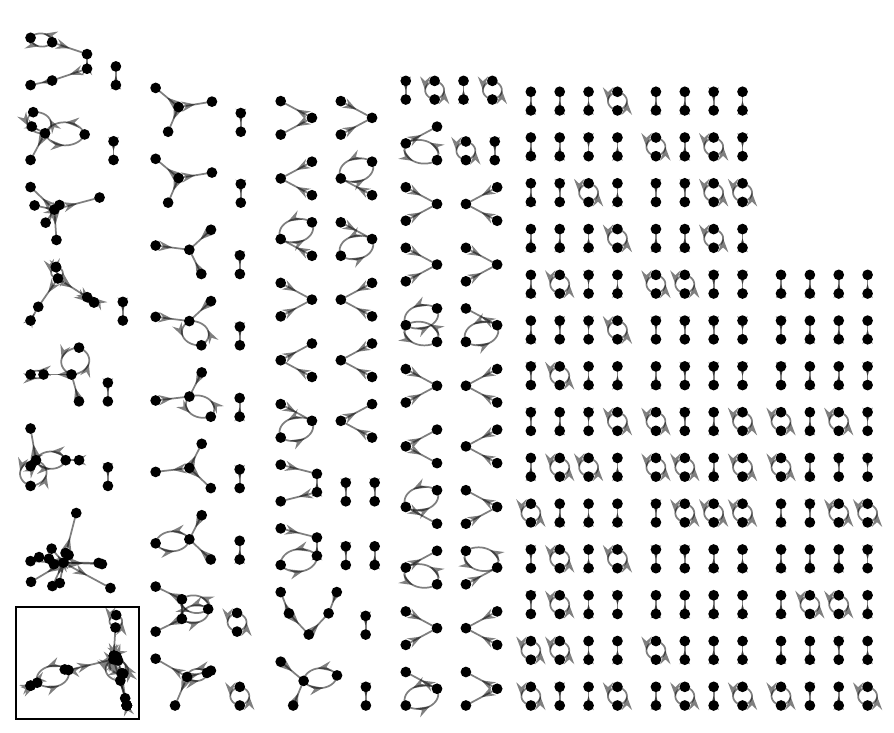}
	\includegraphics[trim = 65mm 105mm 60mm 100mm, clip, width=0.5\textwidth,keepaspectratio]{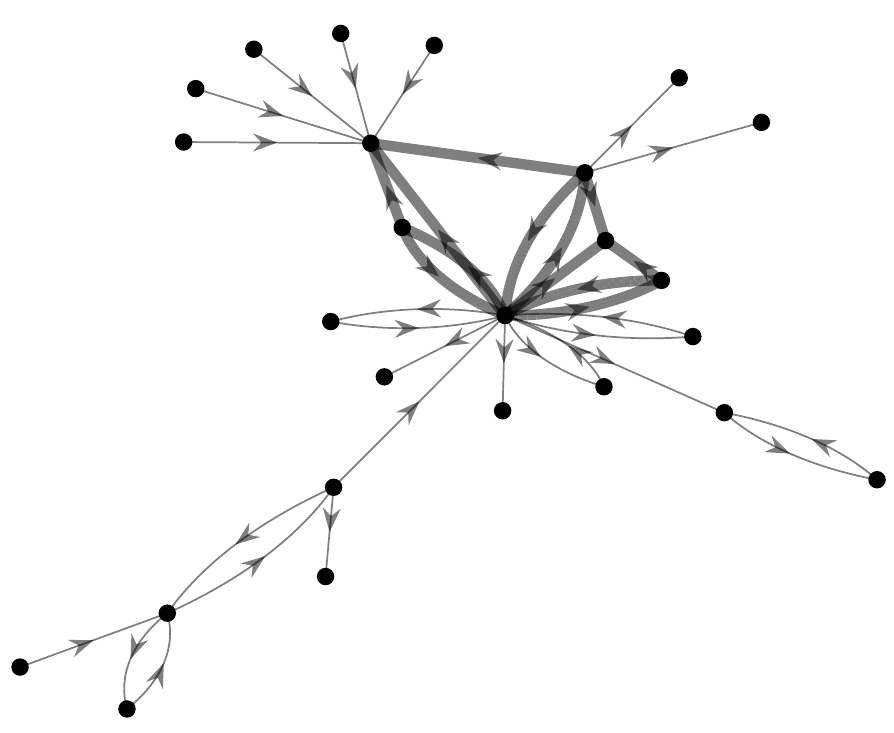}
\caption{ \label{fig:FacebookDigraph} (L) First daily digraph with $\tilde \beta_2 > 0$: day $756$. The only weak component with $\tilde \beta_2 > 0$ is indicated with a box. (R) Detail with (different graph layout and) arcs representing $\tilde H_2$ highlighted. This homology representative is highly symmetrical.
}
\end{figure}

\section*{\label{sec:acknowledgements}Acknowledgements}

The authors thank Michael Robinson for many helpful discussions. This material is based upon work partially supported by the Defense Advanced Research Projects Agency (DARPA) and the Air Force Research Laboratory (AFRL). Any opinions, findings and conclusions or recommendations expressed in this material are those of the author(s) and do not necessarily reflect the views of DARPA or AFRL.


%
%
%
%
%
%
\end{document}